\documentclass[a4paper]{article}

\usepackage{jsaisig}
\usepackage{fancyhdr}
\usepackage{amsthm}
\usepackage{amsmath, amssymb}
\usepackage{graphicx}
\usepackage{enumerate}
\usepackage{algorithm}
\usepackage[noend]{algpseudocode}
\usepackage{url}
\usepackage{tikz}
\usetikzlibrary{positioning}
\usepackage{comment}
\usepackage{secdot}
\sectiondot{subsection}

\newtheoremstyle{jplain}
{}
{}
{\normalfont}
{}
{\bfseries}
{.}
{4pt}
{\thmname{#1} \thmnumber{#2}\thmnote{\hspace{2pt}(#3)}}
\theoremstyle{jplain}

\newtheorem{theorem}{Theorem}
\newtheorem{definition}{Definition}
\newtheorem{lemma}{Lemma}
\newtheorem{example}{Example}
\newtheorem{remark}{Remark}

\makeatletter
\renewenvironment{proof}[1][\proofname]{\par
  \pushQED{\qed}%
  \normalfont \topsep6\p@\@plus6\p@\relax
  \trivlist
  \item\relax
  {\bfseries
  #1\@addpunct{.}}\hspace\labelsep\ignorespaces
}{%
  \popQED\endtrivlist\@endpefalse
}
\makeatother

\fancypagestyle{title}
{
}

\fancypagestyle{normal}
{
 \fancyhead{}
 \fancyfoot[C]{\thepage}
 
}

\begin{document}

\title{The Completeness of Reasoning Algorithms for Clause Sets in Description Logic $\mathcal{ALC}$}

\etitle{}

\author{Daiki Takahashi\afil{1} \quad \quad Ken Kaneiwa\afil{1}}

\affiliation{%
	\afil{1} Department of Computer and Network Engineering, \\
	Graduate School of Informatics and Engineering, The University of Electro-Communications}

\abstract{
On the Semantic Web, metadata and ontologies are used to enable computers to read data.
The Web Ontology Language (OWL) has been proposed as a standard ontological language, and various inference systems for this language have been studied.
Description logics are regarded as the theoretical foundations of OWL; they provide the syntax and semantics of a formal language for describing ontologies and knowledge bases.
In addition, tableau algorithms for description logics have been developed as the standard reasoning algorithms for decidable problems.
However, tableau algorithms generate inefficient reasoning steps owing to their nondeterministic branching for disjunction as well as the increase in the size of models occasioned by existential quantification.
In this study, we propose conjunctive normal form (CNF) concepts, which utilize a flat concept form for description logic $\mathcal{ALC}$ in order to develop algorithms for reasoning about sets of clauses.
We present an efficient reasoning algorithm for clause sets where any $\mathcal{ALC}$ concept is transformed into an equivalent CNF concept.
Theoretically, we prove the soundness, completeness, and termination of the reasoning algorithms for the satisfiability of CNF concepts.
}

\maketitle

\section{Introduction}

The Semantic Web~\cite{DLHandbook_2007} is a framework that enables computers to read information on the Web by adding metadata~\cite{metadata} and ontologies.
It is not possible for computers to distinguish between different metadata, whether they represent the same content or not.
This is solved by defining the relationships between metadata, or by defining metadata in greater detail using a standard vocabulary.
The technology and representation used to realize this is called an ontology~\cite{ontology}.
The Web Ontology Language (OWL)~\cite{OWL_primer_2009} has been proposed by the Web Ontology Working Group of the World Wide Web Consortium (W3C) as a standard language for describing ontologies on the Web.
It is important to formalize the syntax and semantics of an ontological language, based on which the termination, soundness, and completeness of reasoning algorithms can be proven for ontology-based reasoning tasks.

Description logics are logical systems that establish methodologies for knowledge representation and reasoning and that specialize in the characteristics (called concepts) of nominal and adjectival vocabularies.
The languages used in description logics are categorized into families based on expressivity and computational complexity, which depends on the combination of constructs and syntax.
Many studies have been conducted to investigate the properties of description logics with respect to ontological languages and knowledge representation~\cite{DLHandbook_2007}.
Description logics are formalized by syntax, semantics, and reasoning algorithms; the decidability and computational complexity of reasoning tasks have been shown to be the theoretical foundation for OWL.
Description logics are expressive but decidable for application to reasoning, although first\nobreakdash-order logic is not decidable (it is semi-decidable).
Since the tableau method is the standard reasoning algorithm for description logics, the termination, soundness, and completeness of the tableau method are guaranteed.
However, depending on the expressiveness and scale of concepts, branches that increase the number of reasoning operations and the reasoning time may occur, which may lead to inefficient reasoning.

In general, disjunctions and universal or existential quantifiers in logical expressions increase the computational complexity of the reasoning related to the expressions.
For efficient reasoning, logical formulas can be transformed to a conjunctive normal form (CNF)~\cite{CNF} or a clausal form to simplify the derivation.
Clausal forms have been used to simplify logical formulas, and many reasoning algorithms have been proposed for clausal forms.
The Davis\nobreakdash-Putnam\nobreakdash-Logemann\nobreakdash-Loveland (DPLL) algorithm can be applied to the clausal forms of propositional logic to efficiently decide their satisfiability.
The resolution principle is applied to the clausal forms of first-order predicate logic for refutation reasoning.

The standard reasoning algorithms of description logics are based on the tableau calculus~\cite{DLHandbook_2007}.
The hypertableau algorithm~\cite{Hypertableau, HypertableauCalculus} proposed by Boris \textit{et al.} is an improvement on the existing hypertableau algorithm~\cite{Hypertableaux} and hyper\nobreakdash-resolution method~\cite{HyperResolution}.
Although the existing absorption optimizations cannot completely reduce the computational complexity caused by disjunctions, the algorithm by Boris \textit{et al.} solves this problem by improving absorption.
To prove the termination of a reasoning algorithm for knowledge bases containing assertions with relations based on cyclic concepts, such as ``humans who have some human children,'' techniques for blocking a redundant reasoning of individuals with a repetitive set of concepts are applied.
In conventional methods, ancestor pairwise blocking works to prevent model expansion only when one individual is an ancestor of the other.
Anywhere pairwise blocking of the hypertableau algorithm extends this scheme, and avoids mutual blocking by defining an ordering for all individuals while generating blocking between individuals that are not ancestrally related, thereby reducing the computational complexity more efficiently.

The description logic programs~\cite{DLP} have been proposed us an approach that attempts to combine description logics and logic programs.
In this approach, ontology descriptions in description logics and rule expressions in logic programs complement each other by translating through intermediate representations.
As a result, it is possible to reason about description logics using efficient algorithms in logic programs.

In this study, we formalize a reasoning algorithm for the simplified forms of concepts in the description logic $\mathcal{ALC}$ (Attributive Language with Complements)~\cite{ALC_1991} via the following three steps.
\begin{itemize}
    \item A flat concept form (CNF concepts) for description logic $\mathcal{ALC}$ based on the conjunctive normal form is defined.
    \item A decidable reasoning algorithm is designed using inference rules for the clause sets of CNF concepts.
    \item The soundness and completeness of the algorithm for clause sets are proved by introducing restricted tableau for CNF concepts.
\end{itemize}
 Essentially, we define a novel concept form (CNF concepts) in description logics based on CNF, and clause sets of this new concept form are used by reasoning algorithms.
Any $\mathcal{ALC}$ concept can be transformed to an equivalent CNF concept by applying De Morgan's laws, distributive properties, and so on.
CNF concepts can be represented as flat concepts with clause sets, which allows reasoning algorithms to be constructed with efficient inference rules.
The flow of the algorithm is as follows.
First, we simplify the clauses by selecting one literal from each clause in the clause set.
Next, the subconcepts $C$ of the universal role concepts $\forall R.C$ are added as $\exists R.(C \sqcup D)$ to the existential role concepts $\exists R.D$ to reduce all universal role concepts while preserving the equivalence of the concepts.
Finally, the satisfiability of the target concept is determined by determining the satisfiability of the remaining existential role concepts.
We prove the soundness and completeness of the algorithm using a restricted tableau that corresponds to reasoning for clause sets in description logic $\mathcal{ALC}$.

The rest of this paper is structured as follows.
As a preparation, Section 2 presents basic definitions of the syntax and semantics of description logic $\mathcal{ALC}$.
We define the conjunctive normal form in description logics and describe a reasoning algorithm for clause sets in Section 3.
In Section 4, we prove the soundness, completeness, and termination of the algorithm introduced in Section 3.
Some derivation examples using the algorithm are presented in Section 5.
Finally, we conclude the paper and propose future work in Section 6.

\section{Preliminaries}

\subsection{Description Logics}

The concept languages of description logics constitute a family of languages with different expressive power depending on the combination of components and syntactic rules.
A concept language in description logic $\mathcal{ALC}$ consists of a set $\mathbf{CN}$ of concept names $A$, a set $\mathbf{RN}$ of role names $R$, a set $\mathbf{IN}$ of individual names $a$, and the logical connectives $\sqcap$ (conjunction), $\sqcup$ (disjunction), $\lnot$ (negation), and the quantifiers $\exists$ (existential) and $\forall$ (universal).
The top concept $\top$, which contains all individuals, and the bottom concept $\bot$, which contains nothing, are also included in $\mathbf{CN}$.

\begin{definition}[Syntax]
\label{def:ALC_syntax}
    Let $A$ be a concept name, $R$ be a role name, and $C$, $D$ be $\mathcal{ALC}$-concepts.
    The set of $\mathcal{ALC}$-concepts is defined inductively by the following rules: 
    \begin{equation*}
        A \mid \top \mid \bot \mid \lnot C \mid C \sqcap D \mid C \sqcup D \mid \forall R.C \mid \exists R.C
    \end{equation*}
\end{definition}

Complex concepts can be expressed by combining any concept, role names, and logical connections.
As an example, ``animals who have legs" can be represented as follows:
\begin{equation*}
    Animal \sqcap \exists hasPart.Leg
\end{equation*}

An interpretation $\mathcal{I}$ of $\mathcal{ALC}$ consists of a pair $\left(\Delta^\mathcal{I}, \cdot^\mathcal{I}\right)$ of a non-empty set $\Delta^\mathcal{I}$ and a function $\cdot^\mathcal{I}$ such that
\begin{align*}
    \text{for each } A \in \mathbf{CN}, \; & A^\mathcal{I} \subseteq \Delta^\mathcal{I} (\text{in particular}, \; \top^\mathcal{I} = \Delta^\mathcal{I} \text{ and } \bot^\mathcal{I}=\emptyset), \\
    \text{for each } R \in \mathbf{RN}, \; & R^\mathcal{I} \subseteq \Delta^\mathcal{I} \times \Delta^\mathcal{I}, \text{ and} \\
    \text{for each } \circ \in \mathbf{IN}, \; & \circ^\mathcal{I} \in \Delta^\mathcal{I}.
\end{align*}

\begin{definition}[Semantics]
\label{def:ALC_interpretation}
    Let $\mathcal{I} = \left(\Delta^\mathcal{I}, \cdot^\mathcal{I}\right)$ be an interpretation.
    The interpretation of $\mathcal{ALC}$-concepts are defined inductively as follows:
    \begin{align*}
        (\lnot C)^\mathcal{I} & = \Delta^\mathcal{I} \setminus C^\mathcal{I} \\
        (C \sqcap D)^\mathcal{I} & = C^\mathcal{I} \cap D^\mathcal{I} \\
        (C \sqcup D)^\mathcal{I} & = C^\mathcal{I} \cup D^\mathcal{I} \\
        (\forall R.C)^\mathcal{I} & = \left\{x \in \Delta^\mathcal{I} \mid \forall y\left[(x, y) \in R^\mathcal{I} \to y \in C^\mathcal{I}\right]\right\} \\
        (\exists R.C)^\mathcal{I} & = \left\{x \in \Delta^\mathcal{I} \mid \exists y\left[(x, y) \in R^\mathcal{I} \land y \in C^\mathcal{I}\right]\right\}
    \end{align*}
\end{definition}
Let $C$ be an $\mathcal{ALC}$-concept. $C$ is satisfiable if there exists an interpretation $\mathcal{I}$ such that $C^\mathcal{I} \neq \emptyset$, and otherwise, $C$ is unsatisfiable.

\section{Reasoning about Clause Sets in Description Logics}
\label{sec:reason_for_clause_set}

\subsection{Conjunctive Normal Form of Description Logics}
\label{subsec:DL_clause_format}

In this section, we define a new form of $\mathcal{ALC}$-concepts based on the conjunctive normal form.
Every concept name $A$ and the negation $\lnot A$ are concept literals, and if $F$ is a conjunctive normal form, then $\exists R.F$ and $\forall R.F$ are concept literals.
Let us denote a concept literal $L$.
A disjunction of concept literals $L_1 \sqcup \cdots \sqcup L_n$ is called a clause, denoted by $CL$.
A conjunction of clauses $CL_1 \sqcup \cdots \sqcup CL_n$ is called a conjunctive normal form and is denoted as $F$.

\begin{definition}[Concept Literals, Clauses, and Conjunctive Normal Forms]
\label{def:CNF}
    Let $A$ be a concept name, $R$ be a role name, $L_1, \ldots, L_m$ be concept literals, $CL_1, \ldots, CL_n$ be clauses, and $F$ be a conjunctive normal form.
    The set of concept literals, the set of clauses, and the set of conjunctive normal forms are inductively defined as follows:
    \begin{align*}
        L & = A \mid \lnot A  \mid \exists R.F \mid \forall R.F \\
        CL & = L_1 \sqcup \ldots \sqcup L_m \\
        F & = CL_1 \sqcap \ldots \sqcap CL_n \\
    \end{align*}
\end{definition}

Any $\mathcal{ALC}$-concept $C$ can be converted to a conjunctive normal form (denoted by CNF$(C)$), i.e., $C \equiv \mathrm{CNF}\left(C\right)$, by the following laws:
\begin{enumerate}
    \item De Morgan's and double negation laws: the left side is transformed into the right side such that negations appear only in concept names.
    \begin{align*}
        \lnot \left(C \sqcap D \right) & \equiv \lnot C \sqcup \lnot D \\
        \lnot \left(C \sqcup D \right) & \equiv \lnot C \sqcap \lnot D \\
        \lnot \left(\exists R.C\right) & \equiv \forall R.\lnot C \\
        \lnot \left(\forall R.C\right) & \equiv \exists R.\lnot C \\
        \lnot \lnot C & \equiv C
    \end{align*}
    \item Distributive laws: the left side is transformed into the right side such that no conjunction is included in disjunctions.
    \begin{align*}
        \left(C \sqcap D \right) \sqcup E & \equiv \left(C \sqcup E \right) \sqcap \left(D \sqcup E \right) \\
        C \sqcup \left(D \sqcap E \right) & \equiv \left(C \sqcup D \right) \sqcap \left(C \sqcup E \right) \\
    \end{align*}
    \item Associative laws: the parentheses on the left or right side are deleted by transforming into $C \sqcup D \sqcup E$ and $C \sqcap D \sqcap E$, respectively.
    \begin{align*}
        (C \sqcup D) \sqcup E & \equiv C \sqcup (D \sqcup E) \\
        (C \sqcap D) \sqcap E & \equiv C \sqcap (D \sqcap E)
    \end{align*}
\end{enumerate}

A concept name $A$ and its negation $\lnot A$ are complementary literals of each other.
The complementary literals of concept literals $\exists R.F$ and $\forall R.F$ are $\forall R.\mathrm{CNF}(\lnot F)$ and $\exists R.\mathrm{CNF}(\lnot F)$, respectively.
We denote a complementary literal of a concept literal $L$ as $\overline{L}$.

\subsection{Reasoning Algorithms for Clause Sets}

We design a reasoning algorithm for clause sets in $\mathcal{ALC}$.
Given an $\mathcal{ALC}$-concept $C$, the conjunctive normal form $\mathrm{CNF}(C) = CL_1 \sqcap \ldots \sqcap CL_n$ is represented as a clause set as follows:
\begin{equation*}
    \left\{CL_1, \ldots, CL_n\right\}
\end{equation*}
where each clause $CL_i$ is a literal set $\left\{L_1, \ldots, L_n\right\}$ of $L_1 \sqcup \ldots \sqcup L_m$.
In particular, a clause $CL$ is called a unit clause if $|CL| = 1$, and an empty clause if $|CL| = 0$.
Thus, the disjunction $CL_1 \sqcup CL_2$ of two literal sets $CL_1, CL_2$ is represented by $CL_1 \cup CL_2$, and the conjunction $F_1 \sqcap F_2$ of two clause sets $F_1, F_2$ is represented by $F_1 \cup F_2$.
For example, the two clauses $\lnot A_1 \sqcup A_2$ and $A_3 \sqcup \exists R.A_1$ are represented by the literal sets $\{\lnot A_1, A_2\}$ and $\{A_3, \exists R.A_1\}$, respectively.
Furthermore, the conjunctive normal form $(\lnot A_1 \sqcup A_2) \sqcap (A_3 \sqcup \exists R.A_1)$ of those two clauses is represented as the clause set $\{\{\lnot A_1, A_2\}, \{A_3, \exists R.A_1\}\}$.

An algorithm for deciding the satisfiability of the conjunctive normal form of a concept $F = \mathrm{CNF}(C)$ is provided by inference rules.

\begin{definition}[Inference Rules A1, A2, and A3]
\label{def:solve_rule}
    Let $\mathcal{S}_i$ be a family of clause sets, $\mathcal{S}_{i+1}$ is derived from $\mathcal{S}_i$ by applying one of the following rules to each clause set $F \in \mathcal{S}_i$:
    \begin{description}
        \item{$\mathrm{(A1)}$} if $L \in CL$ in $F$, then
        \begin{enumerate}[]
            \item $CL \to \{L\}$,
        \end{enumerate}
        \item{$\mathrm{(A2)}$} if $\forall R.F_1 \in CL$ in $F$, then
        \begin{enumerate}[]
            \item $F \to F \setminus \{CL' \in F \mid \forall R.F_1 \in CL' \}$ and for all $\exists R.F_2 \in CL'$ in $F$, $\exists R.F_2 \to \exists R.(F_1 \cup F_2)$, and
        \end{enumerate}
        \item{$\mathrm{(A3)}$} if all clauses in $F$ are unit clauses of $\{A\}$, $\{\lnot A\}$ or $\{\exists R.F'\}$, then
        \begin{enumerate}[]
            \item $F \to F \setminus \{\{\exists R.F_1\}\}$ and add $F_1$ to  $\mathcal{S}_{i+1}$ for some $\{\exists R.F_1\} \in F$ where $F$ is a parent of $F_1$ with respect to $R$. 
        \end{enumerate}
    \end{description}
\end{definition}

Inference rule A1 is applied to a clause $CL$ containing two or more concept literals.
A concept literal $L$ is selected from $CL$ and all concept literals except $L$ are removed from $CL$.
That is, if $L$ is selected, $CL$ is converted as follows:
\begin{align*}
    L \sqcup L_1 \sqcup \ldots \sqcup L_m \to L \\
\end{align*}
Inference rule A2 is applied to $\forall R.F_1$ in a clause $CL$ in $F$.
After removing all clauses $CL'$ containing the concept $\forall R.F_1$, every existential role concept $\exists R.F_2$ with the same role name $R$ are converted into $\exists R.(F_1 \cup F_2)$ by combining it with $F_1$.
Inference rule A3 is applied to an existential role concept $\exists R.F_2$ if every clause is a unit clause without a universal role concept.
The unit clause $\{\exists R.F_1\}$ is removed from the clause set $F$ and the part $F_1$ of $\exists R.F_1$ is added to $\mathcal{S}_{i+1}$.

Let a derivation tree be a tree such that each node is a family of clause sets $\mathcal{S}_i$ and each edge is an application of an inference rule.
A derivation tree for $F$ is constructed such that the root node is $\mathcal{S}_0 = \left\{F\right\}$ and each child node $\mathcal{S}_{i+1}$ is derived from its parent node $\mathcal{S}_i$ by applying an inference rule.
A family of clause sets $\mathcal{S}_i$ is complete if none of the inference rules can be applied to it in a derivation tree.
Let $\mathcal{S}_i$ be a complete node in a derivation tree for $F$.
If all clause sets in $\mathcal{S}_i$ contain neither empty clauses nor contradictions, then the algorithm decides that $F$ is satisfiable.
Otherwise, it is unsatisfiable.

In a derivation tree, unsatisfiability is inherited by child nodes from their parent nodes.
In other words, if a clause set of $\mathcal{S}_i$ is unsatisfiable, then some clause set is unsatisfiable in $\mathcal{S}_{i+1}$ derived by inference rules A1, A2, and A3.
By contraposition, if all clause sets in $\mathcal{S}_{i+1}$ are satisfiable, then all clause sets in $\mathcal{S}_i$ are also satisfiable.
Therefore, if a node $\mathcal{S}_i$ is satisfiable, then the satisfiability of the root node $\mathcal{S}_0 = \{F\}$ will eventually be decided.

The unsatisfiability is inherited by applying inference rules as follows.
Inference rule A1 deletes disjunctions $L \sqcup L_1 \sqcup \cdots \sqcup L_m$ in a clause.
Since $L \sqsubseteq (L \sqcup L_1 \sqcup \cdots \sqcup L_m)$ hold, if $L \sqcup L_1 \sqcup \cdots \sqcup L_m$ is unsatisfiable, then $L$ is also unsatisfiable.
Inference rule A2 removes some $\forall R.F_1$, and $\exists R.F_2$ is converted to $\exists R.(F_1 \cup F_2)$. 
Since $F \sqsubseteq F \setminus \{CL'\}$ holds, if $F \setminus \{CL'\}$ is unsatisfiable, then $F$ is unsatisfiable.
Further, if $\forall R.F_1 \sqcap \exists R.F_2$ is unsatisfiable, then $F_1 \cup F_2\; (= F_1 \sqcap F_2)$ is unsatisfiable.
Hence, $\exists R.(F_1 \cup F_2) \; \left(= \exists R.(F_1 \sqcap F_2)\right)$ is also unsatisfiable.
By applying inference rule A3, the unit clause of some existential role concept $\{\exists R.F_1\}$ is deleted, and $F_1$ is added to $\mathcal{S}_i$.
If $\exists R.F_1$ is unsatisfiable, then $F_1$ is also unsatisfiable.

\subsection{Efficient Reasoning Algorithm for Clause Sets}

In this section, we improve the reasoning algorithm described in the previous section by exploiting the flat clause representation.
If a concept literal $L$ is selected by inference rule A1, the clause containing $L$ or $\overline{L}$ can be further simplified in the derivation steps.
Moreover, we attempt to avoid the branches caused by applying inference rule A2 to universal role concepts.

\begin{definition}[Inference Rules $\mathrm{A1}^+$ and $\mathrm{A2}^+$]
    The rules A1 and A2 introduced in Definition \ref{def:solve_rule} are replaced by the following inference rules $\mathrm{A1}^+$ and $\mathrm{A2}^+$:
    
    \begin{description}
        \item{$\mathrm{(A1^+)}$} if $L \in CL$ in $F$, then
        \begin{enumerate}[]
            \item for all $CL' \in F$ with $L \in CL'$, $CL' \to \{L\}$, and for all $CL' \in F$ with $\overline{L} \in CL'$, $CL' \to CL' \setminus \{\overline{L}\}$, and
        \end{enumerate}
        \item{$\mathrm{(A2^+)}$} if all clauses in $F$ are unit clauses and $\forall R.F_1 \in CL$ in $F$, then
        \begin{enumerate}[]
            \item $F \to F \setminus \{CL' \in F \mid \forall R.F_1 \in CL' \}$ and for all $\exists R.F_2 \in CL'$ in $F$, $\exists R.F_2 \to \exists R.(F_1 \cup F_2)$.
        \end{enumerate}
    \end{description}
\end{definition}

If a concept literal $L$ is selected from $CL$ by applying inference rule $\mathrm{A1}^+$, then for all clauses $CL'$ containing the literal $L$, the other concept literals are removed from $CL'$, and for all clauses $CL'$ containing the complementary literal $\overline{L}$, the complementary literal is removed from $CL'$.
That is, if $L$ is selected from a clause, all $CL'$ are converted as follows:
\begin{align*}
    L \sqcup L_1 \sqcup \ldots \sqcup L_m \, & \to L \\
    \overline{L} \sqcup L_1 \sqcup \ldots \sqcup L_m \, & \to \, L_1 \sqcup \ldots \sqcup L_m
\end{align*}
Inference rule $\mathrm{A2}^+$ is restricted by the condition that all clauses are unit clauses in order to avoid applying $\mathrm{A2}^+$ before $\mathrm{A1}^+$.

The unsatisfiability is inherited by applying inference rules $\mathrm{A1}^+$ and $\mathrm{A2}^+$.
Inference rule $\mathrm{A1}^+$ deletes disjunctions $L \sqcup L_1 \sqcup \cdots \sqcup L_m$ and $\overline{L} \sqcup L_1 \sqcup \cdots \sqcup L_m$ in all clauses.
Since $L \sqsubseteq (L \sqcup L_1 \sqcup \cdots \sqcup L_m)$ and $(L_1 \sqcup \cdots \sqcup L_m) \sqsubseteq (\overline{L} \sqcup L_1 \sqcup \cdots \sqcup L_m)$ hold, if $L \sqcup L_1 \sqcup \cdots \sqcup L_m$ and $\overline{L} \sqcup L_1 \sqcup \cdots \sqcup L_m$ are unsatisfiable, then $L$ and $L_1 \sqcup \cdots \sqcup L_m$ are also unsatisfiable, respectively.
Furthermore, inference rule $\mathrm{A2}^+$ derives the same clauses as inference rule A2.

\section{Completeness}

We denote by $rol(F)$ a set of all role names contained in a clause set $F$.
For example, $rol(\{\lnot \forall R_1.\exists R_2.C_1, \lnot C_2\}) = \{R_1, R_2\}$.
The set of subexpressions of a non-empty clause set $F$ is the smallest set such that
\begin{enumerate}
    \item $F \in sub(F)$,
    \item if $F' \in sub(F)$, then $F' \subseteq sub(F)$,
    \item if $CL \in sub(F)$, then $CL' \in sub(F)$ for all non-empty $CL' \subseteq CL$,
    \item if $\{\lnot A\} \in sub(F)$, then $\{A\} \in sub(F)$, and
    \item if $\{\forall R.F'\} \in sub(F)$ or $\{\exists R.F'\} \in sub(F)$, then $F' \in sub(F)$.
\end{enumerate}
For example, $sub(\{\{\forall R_1.\left\{\{\lnot A_1\}, \{A_2\}\right\}\}\})$ $=$ $\{\{\{\forall R_1.\left\{\{\lnot A_1\}, \{A_2\}\right\}\}\},$ $\{\forall R_1.\left\{\{\lnot A_1\}, \{A_2\}\right\}\},$ $\left\{\{\lnot A_1\}, \{A_2\}\right\},$ $\{\lnot A_1\}, \{A_2\},$ $\{A_1\}\}$.

\begin{definition}[CNF Tableau]
\label{def:tableau}
    Let $F$ be a clause set.
    A CNF tableau for $F$ is a tuple $T = (S, L, E)$ such that $S$ is a set of individuals, $L: S \to 2^{sub(F)}$ is a function from individuals to subexpressions in $sub(F)$, and $E: rol(F) \to 2^{S \times S}$ is a function from role names to pairs of individuals, and there exists some $s_0 \in S$ such that $F \in L(s_0)$.
    For all $s, t \in S$, it satisfies the following conditions:
    \begin{enumerate}
        \item if $\{L\} \in L(s)$, then $\{\overline{L}\} \notin L(s)$,
        \label{def:tableau_cond1_non_contradiction}
        \item if $F' \in L(s)$, then $F' \subseteq L(s)$,
        \label{def:tableau_cond2_recursive_clause_set}
        \item if $CL \in L(s)$, there exists some $L' \in CL$ such that $\{L'\} \in L(s)$,
        \label{def:tableau_cond3_recursive_literal_set}
        \item if $\{\forall R.F'\} \in L(s)$ and $(s, t) \in E(R)$, then $F' \in L(t)$,
        \label{def:tableau_cond4_universal}
        \item if $\{\exists R.F'\} \in L(s)$, then $(s, t) \in E(R)$, and there exists $t \in S$ such that $F' \in L(t)$, and
        \label{def:tableau_cond5_existential}
        \item if $\{\forall R.F_1\}, \{\exists R.F_2\} \in L(s)$, then $\{\exists R.F_1 \cup F_2\} \in L(s)$.
        \label{def:tableau_cond6_marge}
    \end{enumerate}
\end{definition}

\begin{lemma}
\label{lem:sat_on_exist_tableau}
    There exists a CNF tableau $T = (S, L, E)$ for a clause set $F$ if and only if $F$ is satisfiable.
\end{lemma}

\begin{proof}
    $(\Rightarrow)$ Let $T = (S, L, E)$ be a CNF tableau for a clause set $F$.
    Then, an interpretation $\mathcal{I} = (\Delta^\mathcal{I}, \cdot^\mathcal{I})$ of $F$ can be constructed as follows:
    \begin{gather*}
        \Delta^\mathcal{I} = S, \\
        \text{for each unit clause of } \{A\} \in sub(F) \text{, } A^\mathcal{I} = \{s \in S \mid \{A\} \in L(s)\}, \text{ and} \\
        \text{for each role name } R \in rol(F) \text{, } R^\mathcal{I} = E(R).
    \end{gather*}
    For $\mathcal{I}$, we prove by induction on the structure of a clause set that for all $C$ in $sub(F)$, if $C \in L(s)$, then $s \in C^\mathcal{I}$ holds.
    If $\{A\} \in L(s)$, then by the definition of $\mathcal{I}$, $s \in A^\mathcal{I}$ holds.
    If $\{\lnot A\} \in L(s)$, by Condition (\ref{def:tableau_cond1_non_contradiction}) of Definition \ref{def:tableau}, $\{A\} \notin L(s)$.
    By definition of $\mathcal{I}$, $s \notin A^\mathcal{I}$.
    Hence, $s \in (\lnot A)^\mathcal{I}$.
    If $CL \in L(s)$, then by Condition (\ref{def:tableau_cond3_recursive_literal_set}) of Definition \ref{def:tableau}, there exists some $L' \in CL$ such that $\{L'\} \in L(s)$.
    By the induction hypothesis, $s \in (L')^\mathcal{I}$ holds.
    Thus, $s \in CL^\mathcal{I}$.
    If $\{\forall R.F'\} \in L(s)$, then by Condition (\ref{def:tableau_cond4_universal}) of Definition \ref{def:tableau}, for any $t \in S$, if $(s, t) \in E(R)$, then $F' \in L(t)$.
    By the induction hypothesis, $t \in (F')^\mathcal{I}$ holds.
    Thus, for all $(s, t) \in E(R)$, $(s, t) \in R^\mathcal{I}$, so by Definition \ref{def:ALC_interpretation}, $s \in (\forall R.F')^\mathcal{I}$.
    If $\{\exists R.F'\} \in L(s)$, then by Condition (\ref{def:tableau_cond5_existential}) of Definition \ref{def:tableau}, there exists some $t \in S$ such that $(s, t) \in E(R)$ and $F' \in L(t)$.
    By the induction hypothesis, $t \in (F')^\mathcal{I}$ holds.
    Therefore, since $(s, t) \in R^\mathcal{I}$, by Definition \ref{def:ALC_interpretation}, $s \in (\exists R.F')^\mathcal{I}$.
    If $F' \in L(s)$, then by Condition (\ref{def:tableau_cond2_recursive_clause_set}) of Definition \ref{def:tableau}, $CL \in L(s)$ for all $CL \in F'$.
    By the induction hypothesis, $s \in CL^\mathcal{I}$, so $s \in (F')^\mathcal{I}$.

    Therefore, $F \in sub(F)$, and there exists $s_0$ such that $F \in L(s_0)$ by Definition \ref{def:tableau}, so $s_0 \in F^\mathcal{I}$.
    Hence, the interpretation $\mathcal{I}$ satisfies $F$.

    $(\Leftarrow)$ Assume that an interpretation $\mathcal{I} = (\Delta^\mathcal{I}, \cdot^\mathcal{I})$ satisfies a CNF concept $F$.
    Let us construct the tableau $T = (S, L, E)$ for $F$ such that
    \begin{gather*}
        S = \Delta^\mathcal{I}, \\
        \text{for each } s \in \Delta^\mathcal{I} \text{, } L(s) = \{C \in sub(F) \cup \{\{\exists R.F_1 \cup F_2\} \mid \{\forall R.F_1\}, \{\exists R.F_2\} \in sub(F)\} \mid s \in C^\mathcal{I}\} \text{, and} \\
        \text{for each role name } R \in rol(F) \text{, } E(R) = R^\mathcal{I}.
    \end{gather*}
    Since $F$ is satisfiable, $F^\mathcal{I} \neq \emptyset$.
    That is, there exists some $s \in F^\mathcal{I}$ such that $s \in \Delta^\mathcal{I}$.
    By definition of $T = (S, L, E)$ above, there exists some $s \in S$ such that $s \in F^\mathcal{I}$ and $F \in L(s)$.
    We show that $T$ satisfies the conditions of Definition \ref{def:tableau} for all $s \in S$.
    (\ref{def:tableau_cond1_non_contradiction}) If $\{L\} \in L(s)$ for a concept literal $L$, then $s \in L^\mathcal{I}$ from the construction of $T$.
    Therefore, $s \notin (\overline{L})^\mathcal{I}$, $\{\overline{L}\}\notin L(s)$ from the Definition \ref{def:ALC_interpretation}.
    (\ref{def:tableau_cond2_recursive_clause_set}) For a clause set $F' = \{CL_1, \ldots, CL_n\}$, if $F' \in L(s)$, then $s \in (F')^\mathcal{I}$.
    Thus, for any $CL_i \in F'$, $s \in {CL_i}^\mathcal{I}$, so $CL_i \in L(s)$.
    Therefore, $F' \subseteq L(s)$.
    (\ref{def:tableau_cond3_recursive_literal_set}) If $CL \in L(s)$ for a clause (literal set) $CL = \{L_1, \ldots, L_n\}$, then $s \in CL^\mathcal{I}$.
    Thus, there exists $L' \in CL$ and $s \in (L')^\mathcal{I}$, so $\{L'\} \in L(s)$.
    (\ref{def:tableau_cond4_universal}) If $\{\forall R.F'\} \in L(s)$, then $s \in (\forall R.F')^\mathcal{I}$, so $t \in (F')^\mathcal{I}$ for any $(s, t) \in R^\mathcal{I}$.
    Therefore, $F' \in L(t)$.
    (\ref{def:tableau_cond5_existential}) If $\{\exists R.F'\} \in L(s)$, then $s \in (\exists R.F')^\mathcal{I}$, so there exists some $(s, t) \in R^\mathcal{I}$ and $t \in (F')^\mathcal{I}$.
    Therefore, $F' \in L(t)$.
    (\ref{def:tableau_cond6_marge}) If $\{\forall R.F_1\}, \{\exists R.F_2\} \in L(s)$, from $s \in (\forall R.F_1)^\mathcal{I}, (\exists R.F_2)^\mathcal{I}$, there exists some $(s, t) \in R^\mathcal{I}$ such that $t \in {F_1}^\mathcal{I}, {F_2}^\mathcal{I}$, so $s \in (\exists R.F_1 \cup F_2)^\mathcal{I}$.
    Therefore, $\{\exists R.F_1 \cup F_2\} \in L(s)$. 
    So, $T$ is a CNF tableau for $F$ since $T$ satisfies the conditions of the Definition \ref{def:tableau}.
\end{proof}

Lemma~\ref{lem:sat_on_exist_tableau} shows that an interpretation $\mathcal{I}$ satisfying a clause set $F$ can be constructed from a CNF tableau $T$ for $F$.
Moreover, if a clause set $F$ is satisfiable, a tableau $T$ can be constructed from an interpretation $\mathcal{I}$ of $F$.

We show the termination, soundness, and completeness of the reasoning algorithm for deciding the satisfiability of a clause set.
The termination of the algorithm can be proved by the fact that the applications of inference rules are limited to finite steps.

\begin{theorem}[Termination]
    The reasoning algorithm with A1, A2, and A3 terminates.
    Also, the reasoning algorithm with $\mathrm{A1}^+$, $\mathrm{A2}^+$, and A3 terminates.
\end{theorem}

\begin{proof}
    Let $F$ be a clause set. Since the number of concept literals in $F$ is finite, inference rule A1 (or $\mathrm{A1}^+$) will eventually become inapplicable.
    Inference rules A2 (or $\mathrm{A2}^+$) and A3 delete at least one role name in a clause set, so it terminates in finite steps, depending on the number of role names in $F$.
    Thus, the satisfiability of $F$ can be decided in finite steps.
\end{proof}

\begin{theorem}[Soundness of Reasoning with A1, A2, and A3]
    If the reasoning algorithm with A1, A2, and A3 yields a complete and clash-free family of clause sets $\mathcal{S}_n$ for a clause set $F$, then $F$ is satisfiable.
\end{theorem}

\begin{proof}
    Let $\mathcal{S}_n$ be a complete and clash-free family of clause sets derived from the initial family of clause sets $\mathcal{S}_0 = \{F\}$.
    Then, we prove that a CNF tableau $T = (S, L, E)$ for $F$ can be constructed from the nodes $\mathcal{S}_0, \mathcal{S}_1, \ldots, \mathcal{S}_n$ in a derivation tree where each $\mathcal{S}_{i+1}$ is derived from its parent node $\mathcal{S}_{i}$.

    We show the procedure for constructing a CNF tableau $T = (S, L, E)$ as follows.
    First, we define the set of ordinal numbers of clause sets in $\mathcal{S}_n$.
    \begin{itemize}
        \item $S = \{0, 1, \ldots, |\mathcal{S}_n|\}$
    \end{itemize}
    Let $\mathcal{S}_k = \{F_0, \ldots, F_m\}$ for each $k \leq n$ where each $F_{i+1}$ is added from $F_0, \ldots, F_i$ in the derivation tree (i.e., a clause set of $F_0, \ldots, F_i$ is a parent of $F_{i+1}$).
    Then, the $i$-th clause set in $\mathcal{S}_k$ is defined as follows.
    \begin{equation*}
        \mathcal{S}_k(i) = \left\{
        \begin{array}{cc}
            F_i & \text{if} \quad i \leq m \\
            \emptyset & \text{otherwise}
        \end{array} \right.
    \end{equation*}
    Second, we define the function $L: S \to 2^{sub(F) \cup \{\{\exists R.F_1 \cup F_2\} \mid \{\forall R.F_1\}, \{\exists R.F_2\} \in sub(F)\}}$ from $S$ to clauses, clause sets, and families of clause sets in $\mathcal{S}_0, \ldots, \mathcal{S}_n$ as follows:
    \begin{align*}
        L_F(i) &= \{\mathcal{S}_k(i) \in \{\mathcal{S}_0(i), \ldots, \mathcal{S}_n(i)\} \mid \mathcal{S}_k(i) \neq \emptyset\}, \\
        L_F^+(i) &= \{F' \subseteq \mathcal{S}_k(i) \mid F' \neq \emptyset \text{ and } \mathcal{S}_k(i) \in L_F(i)\}, \\
        L_{CL}(i) &= \{CL' \in \mathcal{S}_k(i) \mid \mathcal{S}_k(i) \in L_F(i)\}, \text{ and } \\
        L(i) &= L_F^+(i) \cup L_{CL}(i) \cup L_{\forall R.F}(i).
    \end{align*}
    where $L_{\forall R.F}(i)$ is the set of unit clauses $\{\forall R.F'\}$ of all universal role concepts $\forall R.F'$ to which inference rule A2 is applied in the $i$-th clause set $\mathcal{S}_k(i)$ for each $k \in \{0, \ldots, n\}$.
    Finally, we define $E: rol(F) \to 2^{S \times S}$ from a parent-child relationship of clause sets in $\mathcal{S}_n$.
    \begin{itemize}
        \item $E(R) = \{(i, j) \in S \times S \mid F_i \text{ is a parent of } F_j \text{ with respect to } R$\}
    \end{itemize}
    
    We show that $T$ satisfies the conditions of Definition \ref{def:tableau}.
    Let $\mathcal{S}_0 = \{F\}$ be the initial family of clause sets.
    By the above definition of $T = (S, L, E)$, there exists $0 \in S$ and $\mathcal{S}_0(0) = F \in L_F^+(0) \subseteq L(0)$.
    Condition (1): If $\{L\} \in L(i)$, then since $\mathcal{S}_n$ is clash-free and also $\mathcal{S}_1, \ldots, \mathcal{S}_{n-1}$ are clash-free, there is no $\mathcal{S}_k(i) \; (\text{for all } 0 \leq k \leq n)$ containing both $\{L\}, \{\overline{L}\}$.
    Therefore, $\{\overline{L}\}$ is not in $L(i)$.
    Condition (2): If $F' \in L(i)$, then $F' \neq \emptyset$ by definition, so $F' \in L_F^+(i)$, and every $CL' \in F' (\subseteq \mathcal{S}_k(i) \in L_F(i))$ is in $L_{CL}(i)$.
    Therefore, $F' \subseteq L(i)$.
    Condition (3): If $CL' \in L(i)$, then $CL' \in L_{CL}(i)$, and there is some $F' \in L_F(i)$ with $CL' \in F'$.
    Since $F' = \mathcal{S}_k(i)$, inference rule A1 (or $\mathrm{A1}^+$) can be applied to derive a unit clause $\{L\}$ of $L \in CL'$ such that $\{L\} \in \mathcal{S}_n(i)$.
    Hence, $\{L\} \in L_{CL}(i) \; (\subseteq L(i))$.
    If inferece rule A1 (or $\mathrm{A1}^+$) is not applied to $CL' \in F'$, then inference rule A2 is applied to some $\forall R.F_1 \in CL'$ with $\{\forall R.F_1\} \in L_{\forall R.F}(i)$.
    Thus, $\{L\} \in L(i)$ for some concept literal $L = \forall R.F_1$ in $CL'$.
    Condition (4): If $(i, j) \in E(R)$, then $\mathcal{S}_{k+1}(j) = F''$ is derived from $\mathcal{S}_k(i)$ (i.e., $\mathcal{S}_k(i)$ is the parent of $F''$) containing $\{\exists R.F''\}$ by inference rule A3.
    If $\{\forall R.F'\} \in L(i)$, then the condition of A3 implies that $\mathcal{S}_k(i)$ does not contain $\{\forall R.F'\}$.
    So, for some $0 \leq k' < k$, $\mathcal{S}_{k'}(i)$ contains a superset of $\{\forall R.F'\}$.
    Since $\{\exists R.F''\}$ in $\mathcal{S}_k(i)$ is derived from $\forall R.F'$ by inference rule A2, $F' \subseteq F''$ holds.
    Therefore, $F' \subseteq F'' \subseteq \mathcal{S}_{k+1}(j)$ and $\mathcal{S}_{k+1}(j) \in L_F(j)$.
    So, $F' \in L_F^+(j)$.
    Hence, $F' \in L(j)$.
    Condition (5): If $\{\exists R.F'\} \in L(i)$, then some $\mathcal{S}_k(i) \in L_F(i)$ contains $\{\exists R.F'\}$ since $\{\exists R.F'\} \in L_{CL}(i)$, and inference rule A3 is applied to $\{\exists R.F'\}$, leading to $\mathcal{S}_{k+1}(j) = F'$.
    Therefore, since $\mathcal{S}_k(i)$ is the parent of $F'$, $(i, j) \in E(R)$ and $F' \in L(j)$ holds.
    Condition (\ref{def:tableau_cond6_marge}): If $\{\forall R.F_1\}, \{\exists R.F_2\} \in L(i)$, then $\{\forall R.F_1\}, \{\exists R.F_2\} \in L_{CL}(i)$.
    So, some $\mathcal{S}_k(i) \in L_F(i)$ contains $\{\forall R.F_1\}, \{\exists R.F_2\}$, and inference rule A2 is applied to derive $\{\exists R.F_1 \cup F_2\} \in \mathcal{S}_{k+1}(i)$ from $\{\exists R.F_2\}$.
    Accordingly, $\{\exists R.F_1 \cup F_2\} \in L_F(i)$.
    Thus, $\{\exists R.F_1 \cup F_2\} \in L_F(i)$ ($\subseteq L(i)$).
    Hence, $T = (S, L, E)$ satisfies the conditions (\ref{def:tableau_cond2_recursive_clause_set})-(\ref{def:tableau_cond6_marge}).

    Therefore, since there exists a CNF tableau for $F$, it is satisfiable by Lemma \ref{lem:sat_on_exist_tableau}.
\end{proof}

\begin{theorem}[Completeness of Derivation with A1, A2, and A3]
\label{thrm:completeness_A1_A2_A3}
    If a clause set $F$ is satisfiable, the reasoning algorithm with A1, A2, and A3 yields a complete and clash-free family of clause sets.
\end{theorem}

\begin{proof}
    Suppose that an interpretation $\mathcal{I} = (\Delta^\mathcal{I}, \cdot^\mathcal{I})$ satisfies $F$.
    By Lemma \ref{lem:sat_on_exist_tableau}, there exists a CNF tableau $T = (S, L, E)$ for $F$ such that $F \in L(s_0)$ for some $s_0 \in S$.
    
    A derivation tree is generated by applying inference rules A1, A2, and A3 to the initial family of clause set $\mathcal{S}_0 = \{F\}$.
    As a result, multiple complete families of clausal sets are derived from the branches by inference rules A1 and A2.
    By selecting a complete concept set $\mathcal{S}_n$ from them, we define the set of ordinal numbers of clausal sets in $S'_{[\mathcal{S}_n]}$.
    \begin{itemize}
        \item $S'_{[\mathcal{S}_n]} = \{0, \ldots, |\mathcal{S}_n|\}$
    \end{itemize}
    We obtain a sequence $\mathcal{S}_0, \mathcal{S}_1, \ldots, \mathcal{S}_n$ from the derivation tree where $\mathcal{S}_0$ is the root node, $\mathcal{S}_n$ is a leat node, and $\mathcal{S}_{i}$ is the parent node of $\mathcal{S}_{i+1}$.
    Let $\mathcal{S}_k = \{F_0, \ldots, F_m\}$ for each $k \leq n$ where each $F_{i+1}$ is added from $F_0, \ldots, F_i$ in the derivation tree (i.e., a clause set of $F_0, \ldots, F_i$ is a parent of $F_{i+1}$).
    Then, the $i$-th clause set in $\mathcal{S}_k$ is defined as follows.
    \begin{equation*}
        \mathcal{S}_k(i) = \left\{
        \begin{array}{cc}
            F_i & \text{if} \quad i \leq m \\
            \emptyset & \text{otherwise}
        \end{array} \right.
    \end{equation*}
    We define a function $\pi: S'_{[\mathcal{S}_n]} \to S$ to connect between $S'_{[\mathcal{S}_n]}$ and $S$ as follows:
    \begin{enumerate}
        \item $\pi(0) = s_0$, and
        \item if $\pi(i) = s$ and $\{\exists R.F'\} \in \mathcal{S}_k(i)$ such that $\mathcal{S}_{k+1}(j) = F'$ is added from $\mathcal{S}_k(i)$ by inference rule A3, then for some $t \in S$ with $F' \in L(t)$ and $(s, t) \in E(R)$, $\pi(j) = t$.
    \end{enumerate}
    To avoid the branching of A1 and A2, we introduce {\it determistic} inference rule A1 that uniquely determines a concept literal $L$ in some $CL \in \mathcal{S}_k(i)$ by following the CNF tableau $T$ with $L \in L(\pi(i))$.
    In addition, A1 is applied in preference to inference rule A2 whenever possible.
    Then, A2 is applied to only to unit clauses.
    It leads to a derivation path as a sequence $\mathcal{S}_0, \mathcal{S}_1, \ldots, \mathcal{S}_n$.
    If $\mathcal{S}_n$ is clash-free, then the reasoning algorithm yields a complete and clash-free family of clause sets.
    We show that $\mathcal{S}_0(i) \cup \cdots \cup \mathcal{S}_n(i) \subseteq L(\pi(i))$.
    Since each $\mathcal{S}_k(i)$ (except for $\mathcal{S}_0(0) = F$) is added by inference rules, all the clauses added to each $\mathcal{S}_k(i)$ have to be included in $L(\pi(i))$.
    We prove this by induction on the depth $k$ of a derivation tree.
    
    If $k = 0$, then $\mathcal{S}_0 = \{F\}$.
    For $\pi(0) = s_0$, $F \in L(\pi(0))$, so $F \subseteq L(\pi(0))$ from Condition (\ref{def:tableau_cond2_recursive_clause_set}) of Definition \ref{def:tableau}.
    Hence, $CL \in L(\pi(0))$ for all $CL \in F (= \mathcal{S}_0(0))$.
    
    If $k > 0$, then inference rules A1, A2, and A3 are applied to $\mathcal{S}_k$.
    
    If deterministic A1 is applied to $CL \in \mathcal{S}_k(i)(\in \mathcal{S}_k)$, then for some $L \in CL$ according to tableau $T$, $\{L\} \in \mathcal{S}_{k+1}(i)$.
    By the induction hypothesis, $CL \in L(\pi(i))$, and so by Condition (\ref{def:tableau_cond3_recursive_literal_set}) of Definition \ref{def:tableau} and the deterministic A1, $\{L\} \in L(\pi(i))$ holds.
    If A2 is applied to $\{\forall R.F_1\}, \{\exists R.F_2\} \in \mathcal{S}_k(i)$, then $\{\exists R.F_2\}$ is replaced with $\{\exists R.F_1 \cup F_2\}$ in $\mathcal{S}_{k+1}$ (i.e., $\{\exists R.F_1 \cup F_2\} \in \mathcal{S}_{k+1}(i)$).
    Since $\{\forall R.F_1\}, \{\exists R.F_2\} \in L(\pi(i))$ by the induction hypothesis, $\{\exists R.F_1 \cup F_2\} \in L(\pi(i))$ by Condition (\ref{def:tableau_cond6_marge}) of Definition \ref{def:tableau}.
    If A3 is applied to $\{\exists R.F_1\} \in \mathcal{S}_k(i)$, then $F_1$ is added to $\mathcal{S}_{k+1}$ (i.e., $\mathcal{S}_{k+1}(j) = F_1$).
    By the induction hypothesis, since $\{\exists R.F_1\} \in L(\pi(i))$, there exists $t \in S$ such that $(\pi(i), t) \in E(R)$ and $F_1 \in L(t)$ from Condition (\ref{def:tableau_cond5_existential}) of Definition \ref{def:tableau}.
    By definition of $\pi$, $\pi(j) = t$, so $F_1 \in L(\pi(j))$.
    By Condition (\ref{def:tableau_cond2_recursive_clause_set}) in Definition \ref{def:tableau}, $F_1 \subseteq L(\pi(j))$.
    Hence, since $\mathcal{S}_0(i) \cup \cdots \cup \mathcal{S}_n(i) \subseteq L(\pi(i))$ for all $i \in S'_{[\mathcal{S}_n]}$, from Condition (\ref{def:tableau_cond1_non_contradiction}) in Definition \ref{def:tableau}, $\mathcal{S}_0(i) \cup \cdots \cup \mathcal{S}_n(i)$ is clash-free.
    Therefore, the complete family of clause sets $\mathcal{S}_n$ is clash-free.
\end{proof}

Next, we show the completeness of the efficient reasoning algorithm with inference rules $\mathrm{A1^+}$, $\mathrm{A2^+}$, and A3.
We define a restricted set of CNF tableaux by revising Condition (\ref{def:tableau_cond3_recursive_literal_set}) of Definition \ref{def:tableau}.

\begin{definition}[Restricted CNF Tableau]
\label{def:restricted_tableau}
    Let $F$ be a clause set $F$.
    A restricted CNF tableau for $F$ is a tuple $T = (S, L, E)$ such that $S$ is a set of individuals, $L: S \to 2^{sub(F)}$ is a function from individuals to subexpressions in $sub(F)$, and $E: rol(F) \to 2^{S \times S}$ is a function from role names to pairs of individuals, and there exists some $s_0 \in S$ such that $F \in L(s_0)$.
    For all $s, t \in S$, it satisfies the following conditions: 
    \begin{enumerate}
        \item if $\{L\} \in L(s)$, then $\{\overline{L}\} \notin L(s)$,
        \label{def:restricted_tableau_cond1_non_contradiction}
        \item if $F' \in L(s)$, then $F' \subseteq L(s)$,
        \label{def:restricted_tableau_cond2_recursive_clause_set}
        \item if $CL \in L(s)$, there exists some $L' \in CL$ such that $\{L'\} \in L(s)$ and for all $CL' \in L(s)$, $CL' \setminus \{\overline{L'}\} \in L(s)$,\label{def:restricted_tableau_cond3_recursive_literal_set}
        \item if $\{\forall R.F'\} \in L(s)$ and $(s, t) \in E(R)$, then $F' \in L(t)$,
        \label{def:restricted_tableau_cond4_universal}
        \item if $\{\exists R.F'\} \in L(s)$, then $(s, t) \in E(R)$, and there exists $t \in S$ such that $F' \in L(t)$, and
        \label{def:restricted_tableau_cond5_existential}
        \item if $\{\forall R.F_1\}, \{\exists R.F_2\} \in L(s)$, then $\{\exists R.F_1 \cup F_2\} \in L(s)$.
        \label{def:restricted_tableau_cond6_marge}
    \end{enumerate}
\end{definition}

\begin{lemma}
\label{lem:exist_restricted_tableau_on_sat}
    There exists a restricted CNF tableau $T = (S, L, E)$ for a clause set $F$ if $F$ is satisfiable.
\end{lemma}

\begin{proof}
    Assume that $F$ is satisfiable.
    Let $L(s)$ be as in the proof of Lemma \ref{lem:sat_on_exist_tableau}.
    Then, $L(s)$ satisfies the conditions in Definition \ref{def:tableau}.
    We prove that $L(s)$ satisfies the Condition (\ref{def:tableau_cond3_recursive_literal_set}) of Definition \ref{def:restricted_tableau}.
    If $CL \in L(s)$, then $s \in CL^\mathcal{I}$, so there exists some $L' \in CL$ and $s \in (L')^\mathcal{I}$.
    Therefore, $\{L'\} \in L(s)$.
    In addition, if $\overline{L'} \notin CL'$ then $CL' \setminus \{\overline{L'}\} = CL' \in L(s)$ for any $CL' \in L(s)$.
    If $\overline{L'} \in CL'$, then $CL' \neq \{\overline{L'}\}$ since $\{\overline{L'}\} \notin L(s)$ from the Condition (\ref{def:restricted_tableau_cond1_non_contradiction}).
    So, $s \notin (\overline{L'})^\mathcal{I}$ since $s \in (L')^\mathcal{I}$.
    Thus, there exists some $L'' \in CL'$ (except for $\overline{L'}$) and $s \in (L'')^\mathcal{I}$.
    Hence, $CL' \setminus \{\overline{L'}\} \in L(s)$ by the definition of $L(s)$ since $s \in (CL' \setminus \{\overline{L'}\})^\mathcal{I}$.
    Therefore, $L(s)$ satisfies the conditions of Definition \ref{def:restricted_tableau}.
\end{proof}

\begin{theorem}[Soundness of Reasoning with $\text{A1}^+$, $\text{A2}^+$, and A3]
    If the reasoning algorithm with $\mathrm{A1}^+$, $\mathrm{A2}^+$, and A3 yields a complete and clash-free family of clause sets $\mathcal{S}_n$ for a clause set $F$, then $F$ is satisfiable.
\end{theorem}

\begin{proof}
    Let $\mathcal{S}_n$ be a complete and clash-free family of clause sets derived from the initial family of clause sets $\mathcal{S}_0 = \{F\}$.
    Then, we prove that a restricted CNF tableau $T = (S, L, E)$ for $F$ can be constructed from the nodes $\mathcal{S}_0, \mathcal{S}_1, \ldots, \mathcal{S}_n$ in a derivation tree where each $\mathcal{S}_{i+1}$ is derived from its parent node $\mathcal{S}_{i}$.

    We show the procedure for constructing a restricted CNF tableau $T = (S, L, E)$ as follows.
    First, we define the set of ordinal numbers of clause sets in $\mathcal{S}_n$.
    \begin{itemize}
        \item $S = \{0, 1, \ldots, |\mathcal{S}_n|\}$
    \end{itemize}
    Let $\mathcal{S}_k = \{F_0, \ldots, F_m\}$ for each $k \leq n$ where each $F_{i+1}$ is added from $F_0, \ldots, F_i$ in the derivation tree (i.e., a clause set of $F_0, \ldots, F_i$ is a parent of $F_{i+1}$).
    Then, the $i$-th clause set in $\mathcal{S}_k$ is defined as follows.
    \begin{equation*}
        \mathcal{S}_k(i) = \left\{
        \begin{array}{cc}
            F_i & \text{if} \quad i \leq m \\
            \emptyset & \text{otherwise}
        \end{array} \right.
    \end{equation*}
    Second, we define the function $L: S \to 2^{sub(F) \cup \{\{\exists R.F_1 \cup F_2\} \mid \{\forall R.F_1\}, \{\exists R.F_2\} \in sub(F)\}}$ from $S$ to clauses, clause sets, and families of clause sets in $\mathcal{S}_0, \ldots, \mathcal{S}_n$ as follows:
    \begin{align*}
        L_F(i) &= \{\mathcal{S}_k(i) \in \{\mathcal{S}_0(i), \ldots, \mathcal{S}_n(i)\} \mid \mathcal{S}_k(i) \neq \emptyset\}, \\
        L_F^+(i) &= \{F' \subseteq \mathcal{S}_k(i) \mid F' \neq \emptyset \text{ and } \mathcal{S}_k(i) \in L_F(i)\}, \\
        L_{CL}(i) &= \{CL' \in \mathcal{S}_k(i) \mid \mathcal{S}_k(i) \in L_F(i)\}, \text{ and } \\
        L(i) &= L_F^+(i) \cup L_{CL}(i) \cup L_{\forall R.F}(i).
    \end{align*}
    where $L_{\forall R.F}(i)$ is the set of unit clauses $\{\forall R.F'\}$ of all universal role concepts $\forall R.F'$ to which inference rule $\mathrm{A2}^+$ is applied in the $i$-th clause set $\mathcal{S}_k(i)$ for each $k \in \{0, \ldots, n\}$.
    Finally, we define $E: rol(F) \to 2^{S \times S}$ from a parent-child relationship of clause sets in $\mathcal{S}_n$.
    \begin{itemize}
        \item $E(R) = \{(i, j) \in S \times S \mid F_i \text{ is a parent of } F_j \text{ with respect to } R$\}
    \end{itemize}
    
    We show that $T$ satisfies the conditions of Definition \ref{def:restricted_tableau}.
    Let $\mathcal{S}_0 = \{F\}$ be the initial family of clause sets.
    By the above definition of $T = (S, L, E)$, there exists $0 \in S$ and $\mathcal{S}_0(0) = F \in L_F^+(0) \subseteq L(0)$.
    Condition (1): If $\{L\} \in L(i)$, then since $\mathcal{S}_n$ is clash-free and also $\mathcal{S}_1, \ldots, \mathcal{S}_{n-1}$ are clash-free, there is no $\mathcal{S}_k(i) \; (\text{for all } 0 \leq k \leq n)$ containing both $\{L\}, \{\overline{L}\}$.
    Therefore, $\{\overline{L}\}$ is not in $L(i)$.
    Condition (2): If $F' \in L(i)$, then $F' \neq \emptyset$ by definition, so $F' \in L_F^+(i)$, and every $CL' \in F' (\subseteq \mathcal{S}_k(i) \in L_F(i))$ is in $L_{CL}(i)$.
    Therefore, $F' \subseteq L(i)$.
    Condition (3): If $CL' \in L(i)$, then $CL' \in L_{CL}(i)$, and there is some $F' \in L_F(i)$ with $CL' \in F'$.
    Since $F' = \mathcal{S}_k(i)$, inference rule $\text{A1}^+$ can be applied to derive a unit clause $\{L\}$ of $L \in CL'$ such that $\{L\} \in \mathcal{S}_n(i)$.
    Hence, $\{L\} \in L_{CL}(i) \; (\subseteq L(i))$.
    In addition, if $L \notin CL'$ for all $CL' \in L(i) \subseteq L_{CL}(i)$, then $\mathcal{S}_{k+1}(i)$ containing $CL' \setminus \{\overline{L}\}$ is derived from $\mathcal{S}_k(i)$ containing $\overline{L} \in CL'$ by inference rule $\mathrm{A1}^+$.
    Hence, $CL' \setminus \{\overline{L}\} \in L(i)$.
    Condition (4): If $(i, j) \in E(R)$, then $\mathcal{S}_{k+1}(j) = F''$ is derived from $\mathcal{S}_k(i)$ (i.e., $\mathcal{S}_k(i)$ is the parent of $F''$) containing $\{\exists R.F''\}$ by inference rule A3.
    If $\{\forall R.F'\} \in L(i)$, then the condition of A3 implies that $\mathcal{S}_k(i)$ does not contain $\{\forall R.F'\}$.
    So, for some $0 \leq k' < k$, $\mathcal{S}_{k'}(i)$ contains $\{\forall R.F'\}$.
    Since $\{\exists R.F''\}$ in $\mathcal{S}_k(i)$ is derived from $\forall R.F'$ by inference rule $\mathrm{A2}^+$, $F' \subseteq F''$ holds.
    Therefore, $F' \subseteq F'' \subseteq \mathcal{S}_{k+1}(j)$ and $\mathcal{S}_{k+1}(j) \in L_F(j)$.
    So, $F' \in L_F^+(j)$.
    Hence, $F' \in L(j)$.
    Condition (5): If $\{\exists R.F'\} \in L(i)$, then some $\mathcal{S}_k(i) \in L_F(i)$ contains $\{\exists R.F'\}$ since $\{\exists R.F'\} \in L_{CL}(i)$, and inference rule A3 is applied to $\{\exists R.F'\}$, leading to $\mathcal{S}_{k+1}(j) = F'$.
    Therefore, since $\mathcal{S}_k(i)$ is the parent of $F'$, $(i, j) \in E(R)$ and $F' \in L(j)$ holds.
    Condition (\ref{def:restricted_tableau_cond6_marge}): If $\{\forall R.F_1\}, \{\exists R.F_2\} \in L(i)$, then $\{\forall R.F_1\}, \{\exists R.F_2\} \in L_{CL}(i)$.
    So, some $\mathcal{S}_k(i) \in L_F(i)$ contains $\{\forall R.F_1\}, \{\exists R.F_2\}$, and inference rule $\mathrm{A2}^+$ is applied to derive $\{\exists R.F_1 \cup F_2\} \in \mathcal{S}_{k+1}(i)$ from $\{\exists R.F_2\}$.
    Accordingly, $\{\exists R.F_1 \cup F_2\} \in L_F(i)$.
    Thus, $\{\exists R.F_1 \cup F_2\} \in L_F(i)$ ($\subseteq L(i)$).
    Hence, $T = (S, L, E)$ satisfies the conditions (\ref{def:restricted_tableau_cond2_recursive_clause_set})-(\ref{def:restricted_tableau_cond6_marge}).

    Therefore, since there exists a restricted CNF tableau for $F$, it is satisfiable by Lemma \ref{lem:exist_restricted_tableau_on_sat}.
\end{proof}

\begin{theorem}[Completeness of Derivation with $\text{A1}^+$, $\text{A2}^+$, and A3]
    If a clause set $F$ is satisfiable, the reasoning algorithm with $\mathrm{A1}^+$, $\mathrm{A2}^+$, and A3 yields a complete and clash-free family of clause sets.
\end{theorem}

\begin{proof}
    Suppose that an interpretation $\mathcal{I} = (\Delta^\mathcal{I}, \cdot^\mathcal{I})$ satisfies $F$.
    By Lemma \ref{lem:exist_restricted_tableau_on_sat}, there exists a restricted CNF tableau $T = (S, L, E)$ for $F$ such that $F \in L(s_0)$ for some $s_0 \in S$.
    
    A derivation tree is generated by applying inference rules $\mathrm{A1}^+$, $\mathrm{A2}^+$, and A3 to the initial family of clause set $\mathcal{S}_0 = \{F\}$.
    As a result, multiple complete families of clausal sets are derived from the branches by inference rule $\mathrm{A1}^+$.
    By selecting a complete concept set $\mathcal{S}_n$ from them, we define the set of ordinal numbers of clausal sets in $S'_{[\mathcal{S}_n]}$.
    \begin{itemize}
        \item $S'_{[\mathcal{S}_n]} = \{0, \ldots, |\mathcal{S}_n|\}$
    \end{itemize}
    We obtain a sequence $\mathcal{S}_0, \mathcal{S}_1, \ldots, \mathcal{S}_n$ from the derivation tree where $\mathcal{S}_0$ is the root node, $\mathcal{S}_n$ is a leat node, and $\mathcal{S}_{i}$ is the parent node of $\mathcal{S}_{i+1}$.
    Let $\mathcal{S}_k = \{F_0, \ldots, F_m\}$ for each $k \leq n$ where each $F_{i+1}$ is added from $F_0, \ldots, F_i$ in the derivation tree (i.e., a clause set of $F_0, \ldots, F_i$ is a parent of $F_{i+1}$).
    Then, the $i$-th clause set in $\mathcal{S}_k$ is defined as follows.
    \begin{equation*}
        \mathcal{S}_k(i) = \left\{
        \begin{array}{cc}
            F_i & \text{if} \quad i \leq m \\
            \emptyset & \text{otherwise}
        \end{array} \right.
    \end{equation*}
    We define a function $\pi: S'_{[\mathcal{S}_n]} \to S$ to connect between $S'_{[\mathcal{S}_n]}$ and $S$ as follows:
    \begin{enumerate}
        \item $\pi(0) = s_0$, and
        \item if $\pi(i) = s$ and $\{\exists R.F'\} \in \mathcal{S}_k(i)$ such that $\mathcal{S}_{k+1}(j) = F'$ is added from $\mathcal{S}_k(i)$ by inference rule A3, then for some $t \in S$ with $F' \in L(t)$ and $(s, t) \in E(R)$, $\pi(j) = t$.
    \end{enumerate}
    To avoid the branching of $\mathrm{A1}^+$, we introduce {\it determistic} inference rule $\mathrm{A1}^+$ that uniquely determines a concept literal $L$ in some $CL \in \mathcal{S}_k(i)$ by following the restricted CNF tableau $T$ with $L \in L(\pi(i))$.
    Inference rule $\mathrm{A2}^+$ only applies to a set of unit clauses, unlike inference rule A2.
    It leads to a derivation path as a sequence $\mathcal{S}_0, \mathcal{S}_1, \ldots, \mathcal{S}_n$.
    If $\mathcal{S}_n$ is clash-free, then the reasoning algorithm yields a complete and clash-free family of clause sets.
    We show that $\mathcal{S}_0(i) \cup \cdots \cup \mathcal{S}_n(i) \subseteq L(\pi(i))$.
    Since each $\mathcal{S}_k(i)$ (except for $\mathcal{S}_0(0) = F$) is added by inference rules, all the clauses added to each $\mathcal{S}_k(i)$ have to be included in $L(\pi(i))$.
    We prove this by induction on the depth $k$ of a derivation tree.
    
    If $k = 0$, then $\mathcal{S}_0 = \{F\}$.
    For $\pi(0) = s_0$, $F \in L(\pi(0))$, so $F \subseteq L(\pi(0))$ from Condition (\ref{def:restricted_tableau_cond2_recursive_clause_set}) of Definition \ref{def:restricted_tableau}.
    Hence, $CL \in L(\pi(0))$ for all $CL \in F (= \mathcal{S}_0(0))$.
    
    If $k > 0$, then inference rules $\mathrm{A1}^+$, $\mathrm{A2}^+$, and A3 are applied to $\mathcal{S}_k$.
    
    If deterministic $\mathrm{A1}^+$ is applied to $CL \in \mathcal{S}_k(i)(\in \mathcal{S}_k)$, then for some $L \in CL$ according to restricted CNF tableau $T$, $\{L\} \in \mathcal{S}_{k+1}(i)$.
    By the induction hypothesis, $CL \in L(\pi(i))$, and so by Condition (\ref{def:restricted_tableau_cond3_recursive_literal_set}) of Definition \ref{def:restricted_tableau} and the deterministic $\mathrm{A1}^+$, $\{L\} \in L(\pi(i))$ holds.
    In addition, $CL' \setminus \{\overline{L}\} \in \mathcal{S}_{k+1}(i)$ for all $CL' \in \mathcal{S}_k(i)$ such that $L \notin CL'$.
    By the induction hypothesis, $CL' \in L(\pi(i))$.
    By Condition (\ref{def:restricted_tableau_cond3_recursive_literal_set}) of Definition \ref{def:restricted_tableau}, $CL' \setminus \{\overline{L}\} \in L(\pi(i))$.
    
    If $\mathrm{A2}^+$ and A3 are applied, they are proved as well as Theorem \ref{thrm:completeness_A1_A2_A3}.
    Hence, since $\mathcal{S}_0(i) \cup \cdots \cup \mathcal{S}_n(i) \subseteq L(\pi(i))$ for all $i \in S'_{[\mathcal{S}_n]}$, from Condition (\ref{def:restricted_tableau_cond1_non_contradiction}) in Definition \ref{def:restricted_tableau}, $\mathcal{S}_0(i) \cup \cdots \cup \mathcal{S}_n(i)$ is clash-free.
    Therefore, the complete family of clause sets $\mathcal{S}_n$ is clash-free.
\end{proof}

\section{Example of Reasoning}

We provide some examples concerning deciding the satisfiability of a CNF concept using the reasoning algorithm with inference rules A1, A2, and A3.
In addition, we present an example of efficient derivation by the reasoning algorithm with inference rules $\mathrm{A1}^+$, $\mathrm{A2}^+$, and A3.

\begin{example}[CNF Concepts] We consider the following $\mathcal{ALC}$ concept $F$:
    \begin{align*}
        F = &(Animal \sqcup (Black \sqcap \forall hasPart.Small)) \sqcap \\
        &(\lnot Animal \sqcup \exists hasPart.(Leg \sqcap \lnot Small)) \sqcap \lnot (\exists hasPart.Leg \sqcap \exists hasPart.Wing)
    \end{align*}
    
    As explained in Section \ref{sec:reason_for_clause_set}, the concept $F$ is transformed into the conjunctive normal form as follows:
    \begin{align*}
        \mathrm{CNF}(F) = &(Animal \sqcup Black) \sqcap (Animal \sqcup \forall hasPart.Small) \sqcap \\
        &(\lnot Animal \sqcup \exists hasPart.(Leg \sqcap \lnot Small)) \sqcap (\forall hasPart.\lnot Leg \sqcup \forall hasPart.\lnot Wing)
    \end{align*}
\end{example}

\begin{example}[Derivation with A1, A2, and A3]
    Let $A$, $B$, $S$, $L$, $W$, and $h$ be concept names $Animal$, $Black$, $Small$, $Leg$, $Wing$, and a role name $hasPart$, respectively.
    A derivation tree from the initial family of clausal sets $\mathcal{S}_0 = \{\mathrm{CNF}(F)\}$ to $\mathcal{S}_{10}$ is shown in Figure \ref{fig:solve_example_A1_A2}.

    For the first clause $Animal \sqcup Black~(= \{A, B\})$ of the concept $\mathrm{CNF}(F) \in \mathcal{S}_0$, $\mathcal{S}_1 = \{F_1\}$ is derived by applying inference rule A1 for the concept literal $Animal(= A)$ (ex-1).
    \begin{align*}
        F_1 = &Animal \sqcap (Animal \sqcup \forall hasPart.Small) \sqcap \\
        &(\lnot Animal \sqcup \exists hasPart.(Leg \sqcap \lnot Small)) \sqcap (\forall hasPart.\lnot Leg \sqcup \forall hasPart.\lnot Wing)
    \end{align*}
    
    For the second clause $Animal \sqcup \forall hasPart.Small~(= \{A, \forall h.\{\{S\}\}\})$ of $F_1 \in \mathcal{S}_1$, $\mathcal{S}_2 = \{F_2\}$ is derived by applying rule A1 for $Animal~(= A)$ (ex-2).
    \begin{align*}
        F_2 &= Animal \sqcap Animal \sqcap (\lnot Animal \sqcup \exists hasPart.(Leg \sqcap \lnot Small)) \sqcap (\forall hasPart.\lnot Leg \sqcup \forall hasPart.\lnot Wing) \\
        &= Animal \sqcap (\lnot Animal \sqcup \exists hasPart.(Leg \sqcap \lnot Small)) \sqcap (\forall hasPart.\lnot Leg \sqcup \forall hasPart.\lnot Wing)
    \end{align*}
    
    For the second clause $\lnot Animal \sqcup \exists hasPart.(Leg \sqcap \lnot Small)~(= \{\lnot A, \exists h.\{\{L\}, \{\lnot S\}\}\})$ of $F_2 \in \mathcal{S}_2$, $\mathcal{S}_3 = \{F_3\}$ is derived by applying inference rule A1 for $\lnot Animal~(= \lnot A)$ (ex-3).
    \begin{align*}
        F_3 = Animal \sqcap \lnot Animal \sqcap (\forall hasPart.\lnot Leg \sqcup \forall hasPart.\lnot Wing)
    \end{align*}
    
    In this case, $F_3$ is unsatisfiable because it contains a clash (i.e., $Animal$ and its negation $ \lnot Animal$).
    Hence, another concept literal occurring in $\mathcal{S}_2$ is selected in an application of inference rule A1.
    That is, for the second clause $\lnot Animal \sqcup \exists hasPart.(Leg \sqcap \lnot Small)~(= \{\lnot A, \exists h.\{\{L\}, \{\lnot S\}\}\})$ of $F_2 \in \mathcal{S}_2$, $\mathcal{S}_4 = \{F_4\}$ is derived by applying inference rule A1 for the concept literal $\exists hasPart.(Leg \sqcap \lnot Small)~(= \exists h.\{\{L\}, \{\lnot S\}\})$ (ex-4).
    \begin{align*}
        F_4 &= Animal \sqcap \exists hasPart.(Leg \sqcap \lnot Small) \sqcap (\forall hasPart.\lnot Leg \sqcup \forall hasPart.\lnot Wing)
    \end{align*}
    
    For the third clause $\forall hasPart.\lnot Leg \sqcup \forall hasPart.\lnot Wing~(= \{\forall h.\{\{\lnot L\}\}, \forall h.\{\{\lnot W\}\}\})$ of $F_4 \in \mathcal{S}_4$, $\mathcal{S}_5 = \{F_5\}$ is derived by applying inference rule A1 for the concept literal $\forall hasPart.\lnot Leg~(= \forall h.\{\{\lnot L\}\})$ (ex-5).
    \begin{equation*}
        F_5 = Animal \sqcap \exists hasPart.(Leg \sqcap \lnot Small) \sqcap \forall hasPart.\lnot Leg
    \end{equation*}
    
    For the third clause $\forall hasPart.\lnot Leg~(= \forall h.\{\{\lnot L\}\})$ of $F_5 \in \mathcal{S}_5$, $\mathcal{S}_6 = \{F_6\}$ is derived by applying inference rule A2 for all existential role concepts (i.e., $\exists hasPart.(\lnot Leg \sqcap Leg \sqcap \lnot Small)$) in $F_5$ (ex-6).
    \begin{equation*}
        F_6 = Animal \sqcap \exists hasPart.(\lnot Leg \sqcap Leg \sqcap \lnot Small)
    \end{equation*}

    For the second clause $\exists hasPart.(\lnot Leg \sqcap Leg \sqcap \lnot Small)~(= \{\exists h.\{\{\lnot L\}, \{L\}, \{\lnot S\}\}\} )$ in $F_6$ of $\mathcal{S}_6$, $\mathcal{S}_7 = \{F_7, F_8\}$ is derived by applying inference rule A3 for $F_6$ where $F_6$ is the parent of $F_8$ (ex-7).
    \begin{align*}
        F_7 &= Animal \\
        F_8 &= \lnot Leg \sqcup Leg \sqcap \lnot Small
    \end{align*}
    
    In this case, $F_8$ is unsatisfiable because it contains a clash (i.e., $Leg$ and its negation $\lnot Leg$).
    So, another concept literal occurring in $\mathcal{S}_4$ is selected in an application of inference rule A1.
    That is, for the third clause $\forall hasPart.\lnot Leg \sqcup \forall hasPart.\lnot Wing~(= \{\forall h.\{\{\lnot L\}\},$ $\forall h.\{\{\lnot W\}\}\})$ of $F_4 \in \mathcal{S}_4$, $\mathcal{S}_8 = \{F_9\}$ is derived by applying rule A1 for the concept literal $\forall hasPart.\lnot Wing~(= \forall h.\{\{\lnot W\}\})$ (ex-8).
    \begin{align*}
        F_9 &= Animal \sqcap \exists hasPart.(Leg \sqcap \lnot Small) \sqcap \forall hasPart.\lnot Wing
    \end{align*}
    
    For the third clause $\forall hasPart.\lnot Wing~(= \{\forall h.\{\{\lnot W\}\}\})$ of $F_9 \in \mathcal{S}_8$, $\mathcal{S}_9 = \{F_{10}\}$ is derived by applying rule A2 for all existential role concepts (i.e., $\exists hasPart.(Leg \sqcap \lnot Small)~(= \exists h.\{ \{L\}, \{\lnot S\}\})$) (ex-9).
    \begin{align*}
        F_{10} &= Animal \sqcap \exists hasPart.(\lnot Wing \sqcap Leg \sqcap \lnot Small)
    \end{align*}
    
    For the second clause $\exists hasPart.(\lnot Wing \sqcap Leg \sqcap \lnot Small)~(= \exists h.\{\{\lnot W\}, \{L\}, \{\lnot S\}\})$ of $F_{10} \in \mathcal{S}_9$, $\mathcal{S}_{10} = \{F_{11}, F_{12}\}$ is derived by applying inference rule A3 for $F_{10}$ where $F_{10}$ is the parent of $F_{12}$ (ex-10).
    \begin{align*}
        F_{11} &= Animal \\
        F_{12} &= \lnot Wing \sqcap Leg \sqcap \lnot Small
    \end{align*}
    
    In this case, $\mathcal{S}_{10} = \{F_{11}, F_{12}\}$ is complete and clash-free.
    So we can decide that the concept $F$ is satisfiable.
    
    \begin{figure}[htbp]
        \centering
        \begin{tikzpicture}[node/.style={circle,draw}]
            \node [node] at (0, 0) (S0) {$\mathcal{S}_0$};
            \node [right=0cm of S0] (S0_t) {$= \{\{\{A, B\}, \{A, \forall h.\{\{S\}\}\}, \{\lnot A, \exists h.\{\{L\}, \{\lnot S\}\}\}, \{\forall h.\{\{\lnot L\}\}, \forall h.\{\{\lnot W\}\}\}\}\}$};
            \node [node] at (0, -1.5) (S1) {$\mathcal{S}_1$};
            \node [right=0cm of S1] (S1_t) {$= \{\{\underline{\{A\}}, \{A, \forall h.\{\{S\}\}\}, \{\lnot A, \exists h.\{\{L\}, \{\lnot S\}\}\}, \{\forall h.\{\{\lnot L\}\}, \forall h.\{\{\lnot W\}\}\}\}\}$};
            \draw[->] (S0) -- (S1) node [midway, right] {A1 (ex-1)};
            \node [node] at (0, -3) (S2) {$\mathcal{S}_2$};
            \node [right=0cm of S2] (S2_t) {$= \{\{\underline{\{A\}}, \{\lnot A, \exists h.\{\{L\}, \{\lnot S\}\}\}, \{\forall h.\{\{\lnot L\}\}, \forall h.\{\{\lnot W\}\}\}\}\}$};
            \draw[->] (S1) -- (S2) node [midway, right] {A1 (ex-2)};
            \node [node] at (0.5, -4.5) (S3) {$\mathcal{S}_3$};
            \node [right=0cm of S3] (S3_t) {$= \{\{\{A\}, \underline{\{\lnot A\}}, \{\forall h.\{\{\lnot L\}\}, \forall h.\{\{\lnot W\}\}\}\}\}$};
            \node [below=0cm of S3] (S3_clash) {clash};
            \draw[->] (S2) -- (S3) node [midway, right] {A1 (ex-3)};
            \node [node] at (0, -7) (S4) {$\mathcal{S}_4$};
            \node [right=0cm of S4] (S4_t) {$= \{\{\{A\}, \underline{\{\exists h.\{\{L\}, \{\lnot S\}\}\}}, \{\forall h.\{\{\lnot L\}\}, \forall h.\{\{\lnot W\}\}\}\}\}$};
            \draw[->] (S2) to[bend right=30] node[pos=0.8, right] {A1 (ex-4)} (S4);
            \node [node] at (0.5, -8.5) (S5) {$\mathcal{S}_5$};
            \node [right=0cm of S5] (S5_t) {$= \{\{\{A\}, \{\exists h.\{\{L\}, \{\lnot S\}\}\}, \underline{\{\forall h.\{\{\lnot L\}\}\}}\}\}$};
            \draw[->] (S4) -- (S5) node [midway, right] {A1 (ex-5)};
            \node [node] at (0.5, -10) (S6) {$\mathcal{S}_6$};
            \node [right=0cm of S6] (S6_t) {$= \{\{\{A\}, \underline{\{\exists h.\{\{\lnot L\}, \{L\}, \{\lnot S\}\}\}}\}\}$};
            \draw[->] (S5) -- (S6) node [midway, right] {A2 (ex-6)};
            \node [node] at (0.5, -11.5) (S7) {$\mathcal{S}_7$};
            \node [right=0cm of S7] (S7_t) {$= \{\{\{A\}\}, \underline{\{\{\lnot L\}, \{L\}, \{\lnot S\}\}}\}$};
            \node [below=0cm of S7] (S7_clash) {clash};
            \draw[->] (S6) -- (S7) node [midway, right] {A3 (ex-7)};
            \node [node] at (0, -14) (S8) {$\mathcal{S}_8$};
            \node [right=0cm of S8] (S8_t) {$= \{\{\{A\}, \{\exists h.\{\{L\}, \{\lnot S\}\}\}, \underline{\{\forall h.\{\{\lnot W\}\}\}}\}\}$};
            \draw[->] (S4) to[bend right=20] node[pos=0.9, right] {A1 (ex-8)} (S8);
            \node [node] at (0, -15.5) (S9) {$\mathcal{S}_9$};
            \node [right=0cm of S9] (S9_t) {$= \{\{\{A\}, \underline{\{\exists h.\{\{\lnot W\}, \{L\}, \{\lnot S\}\}\}}\}\}$};
            \draw[->] (S8) -- (S9) node [midway, right] {A2 (ex-9)};
            \node [node] at (0, -17) (S10) {$\mathcal{S}_{10}$};
            \node [right=0cm of S10] (S10_t) {$= \{\{\{A\}\}, \underline{\{\{\lnot W\}, \{L\}, \{\lnot S\}\}}\}$};
            \node [below=0cm of S10] (S10_SAT) {SAT};
            \draw[->] (S9) -- (S10) node [midway, right] {A3 (ex-10)};
        \end{tikzpicture}
        \caption{A derivation tree with A1, A2, and A3}
        \label{fig:solve_example_A1_A2}
    \end{figure}
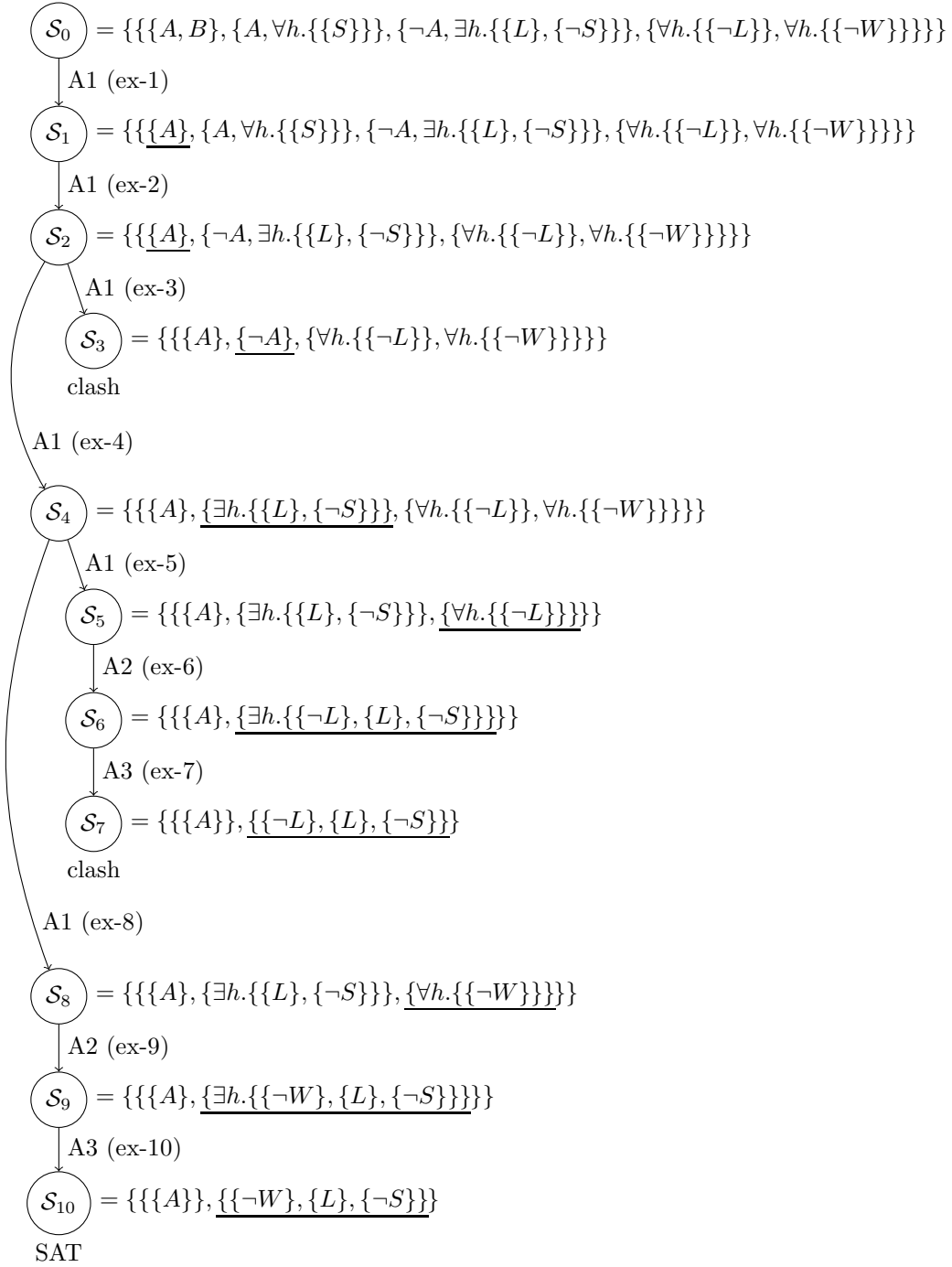
\end{example}

\begin{example}[Derivation with $\text{A1}^+$, $\text{A2}^+$, and A3]

    A derivation tree from the initial family of clause sets ${\mathcal{S}_0}' = \{\mathrm{CNF}(F)\}$ to ${\mathcal{S}_5}'$ is shown in Figure \ref{fig:solve_example_A1+}.
    For the first clause of the concept $\mathrm{CNF}(F) \in {\mathcal{S}_0}'$, ${\mathcal{S}_1}' = \{{F_1}'\}$ (ex-1') is derived by applying inference rule $\mathrm{A1}^+$ for the concept literal $Animal~(= \{A\})$.
    \begin{align*}
        {F_1}' &= Animal \sqcap Animal \sqcap \exists hasPart.(Leg \sqcap \lnot Small) \sqcap (\forall hasPart.\lnot Leg \sqcup \forall hasPart.\lnot Wing) \\
        &= Animal \sqcap \exists hasPart.(Leg \sqcap \lnot Small) \sqcap (\forall hasPart.\lnot Leg \sqcup \forall hasPart.\lnot Wing)
    \end{align*}
    
    For the third clause $\forall hasPart.\lnot Leg \sqcup \forall hasPart.\lnot Wing~(=\{\forall h.\{\{\lnot L\}\}, \forall h.\{\{\lnot W\}\}\})$ of ${F_1}' \in {\mathcal{S}_1}'$, ${\mathcal{S}_2}' = \{{F_2}'\}$ is derived by applying inference rule $\mathrm{A1}^+$ for the concept literal $\forall hasPart.\lnot Leg~(= \forall h.\{\{\lnot L\}\})$ (ex-2').
    \begin{align*}
        {F_2}' = Animal \sqcap \exists hasPart.(Leg \sqcap \lnot Small) \sqcap \forall hasPart.\lnot Leg
    \end{align*}
    
    For the third clause $\forall hasPart.\lnot Leg~(= \{\forall h.\{\{\lnot L\}\}\})$ of ${F_2}' \in {\mathcal{S}_2}'$, ${\mathcal{S}_3}' = \{{F_3}'\}$ is derived by applying inference rule $\mathrm{A2}^+$ for all existential role concepts (i.e., $\exists hasPart.(Leg \sqcap \lnot Small)~(=\{\exists h.\{ \{L\}, \{\lnot S\}\}\})$) (ex-3').
    \begin{equation*}
        {F_3}' = Animal \sqcap \exists hasPart.(\lnot Leg \sqcap Leg \sqcap \lnot Small)
    \end{equation*}
    
    For the second clause $\exists hasPart.(\lnot Leg \sqcap Leg \sqcap \lnot Small)~(= \{\exists h.\{\{\lnot L\}, \{L\}, \{\lnot S\}\}\})$ of ${F_3}' \in {\mathcal{S}_3}'$, ${\mathcal{S}_4}' = \{{F_4}', {F_5}'\}$ is derived by applying inference rule A3 for ${F_3}'$ where $F_{3}'$ is the parent of $F_{5}'$ (ex-4').
    \begin{align*}
        {F_4}' &= Animal \\
        {F_5}' &= \lnot Leg \sqcup Leg \sqcap \lnot Small
    \end{align*}
    
    In this case, ${F_5}'$ is unsatisfiable because it contains a clash (i.e., $Leg$ and its negation $\lnot Leg$).
    Therefore, another concept literal occurring in $\mathcal{S}_1'$ is selected in an application of inference rule $\mathrm{A1}^+$ (ex-2').
    That is, for the third clause $\forall hasPart.\lnot Leg \sqcup \forall hasPart.\lnot Wing~(=\{\forall h.\{\{\lnot L\}\}, \forall h.\{\{\lnot W\}\}\})$ of ${F_1}' \in {\mathcal{S}_1}'$, ${\mathcal{S}_5}' = \{{F_6}'\}$ is derived by applying inference rule $\mathrm{A1}^+$ for the concept literal $\forall hasPart.\lnot Wing~(= \forall h.\{\{\lnot W\}\})$ (ex-5').
    \begin{align*}
        {F_6}' = Animal \sqcap \exists hasPart.(Leg \sqcap \lnot Small) \sqcap \forall hasPart.\lnot Leg
    \end{align*}

    For the third clause $\forall hasPart.\lnot Wing~(= \{\forall h.\{\{\lnot W\}\}\})$ of ${F_6}' \in {\mathcal{S}_5}'$, ${\mathcal{S}_6}' = \{{F_7}'\}$ is derived by applying inference rule $\mathrm{A2}^+$ for all existential role concepts (i.e., $\exists hasPart.(Leg \sqcap \lnot Small)~(=\{\exists h.\{ \{L\}, \{\lnot S\}\}\})$) (ex-6').
    \begin{equation*}
        {F_7}' = Animal \sqcap \exists hasPart.(\lnot Wing \sqcap Leg \sqcap \lnot Small)
    \end{equation*}
    
    For the second clause $\exists hasPart.(\lnot Wing \sqcap Leg \sqcap \lnot Small)~(= \{\exists h.\{\{\lnot W\}, \{L\}, \{\lnot S\}\}\})$ of ${F_7}' \in {\mathcal{S}_6}'$, ${\mathcal{S}_7}' = \{{F_8}', {F_9}'\}$ is derived by applying inference rule A3 for ${F_7}'$ where $F_{7}'$ is the parent of $F_{9}'$ (ex-7').
    \begin{align*}
        {F_8}' &= Animal \\
        {F_9}' &= \lnot Wing \sqcap Leg \sqcap \lnot Small
    \end{align*}
    
    In this case, ${\mathcal{S}_7}' = \{{F_8}', {F_9}'\}$ is complete and clash-free.
    So we can decide that the concept $F$ is satisfiable.
\end{example}
    
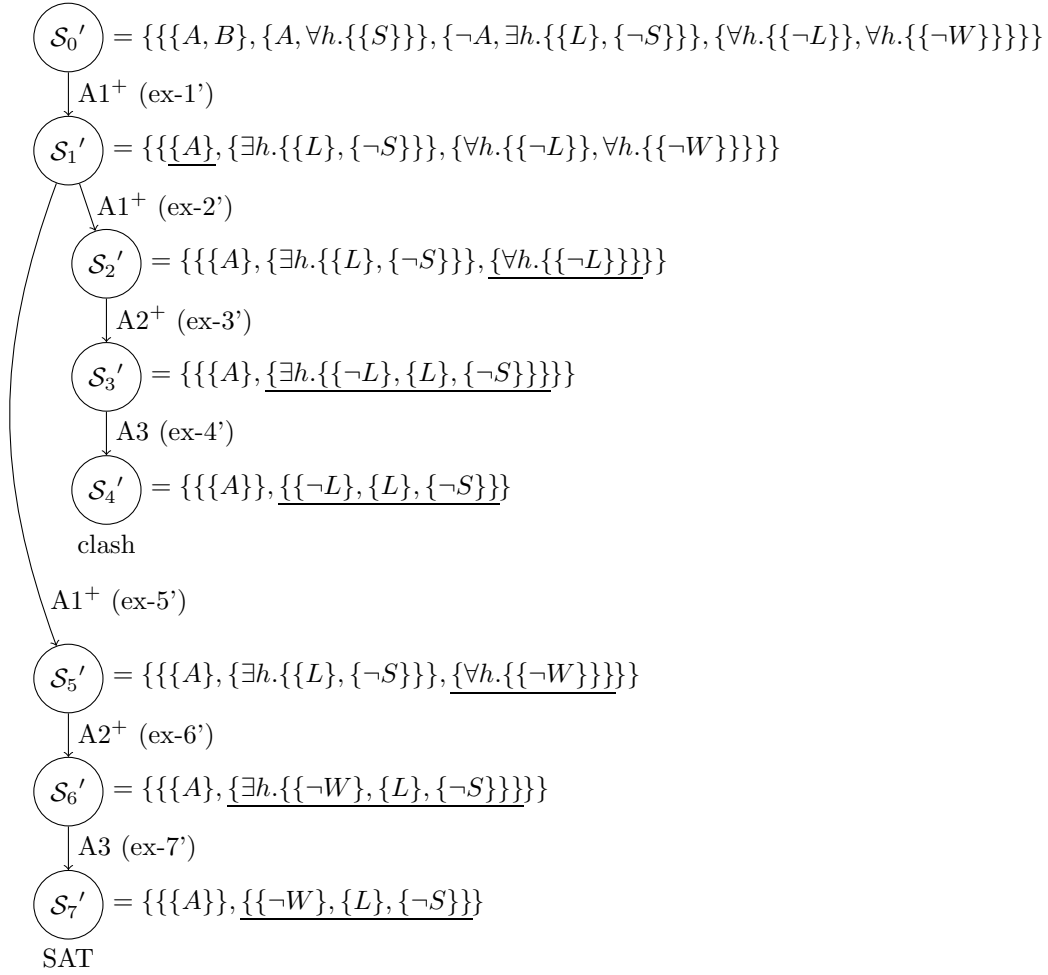
\begin{figure}[htb]
    \centering
    \begin{tikzpicture}[node/.style={circle,draw}]
        \node [node] at (0, 0) (S0) {${\mathcal{S}_0}'$};
        \node [right=0cm of S0] (S0_t) {$= \{\{\{A, B\}, \{A, \forall h.\{\{S\}\}\}, \{\lnot A, \exists h.\{\{L\}, \{\lnot S\}\}\}, \{\forall h.\{\{\lnot L\}\}, \forall h.\{\{\lnot W\}\}\}\}\}$};
        \node [node] at (0, -1.5) (S1) {${\mathcal{S}_1}'$};
        \node [right=0cm of S1] (S1_t) {$= \{\{\underline{\{A\}}, \{\exists h.\{\{L\}, \{\lnot S\}\}\}, \{\forall h.\{\{\lnot L\}\}, \forall h.\{\{\lnot W\}\}\}\}\}$};
        \draw[->] (S0) -- (S1) node [midway, right] {$\mathrm{A1}^+$ (ex-1')};
        \node [node] at (0.5, -3) (S2) {${\mathcal{S}_2}'$};
        \node [right=0cm of S2] (S2_t) {$= \{\{\{A\}, \{\exists h.\{\{L\}, \{\lnot S\}\}\}, \underline{\{\forall h.\{\{\lnot L\}\}\}}\}\}$};
        \draw[->] (S1) -- (S2) node [midway, right] {$\mathrm{A1}^+$ (ex-2')};
        \node [node] at (0.5, -4.5) (S3) {${\mathcal{S}_3}'$};
        \node [right=0cm of S3] (S3_t) {$= \{\{\{A\}, \underline{\{\exists h.\{\{\lnot L\}, \{L\}, \{\lnot S\}\}\}}\}\}$};
        \draw[->] (S2) -- (S3) node [midway, right] {$\mathrm{A2}^+$ (ex-3')};
        \node [node] at (0.5, -6) (S4) {${\mathcal{S}_4}'$};
        \node [right=0cm of S4] (S4_t) {$= \{\{\{A\}\}, \underline{\{\{\lnot L\}, \{L\}, \{\lnot S\}\}}\}$};
        \node [below=0cm of S4] (S4_clash) {clash};
        \draw[->] (S3) -- (S4) node [midway, right] {A3 (ex-4')};
        \node [node] at (0, -8.5) (S5) {${\mathcal{S}_5}'$};
        \node [right=0cm of S5] (S5_t) {$= \{\{\{A\}, \{\exists h.\{\{L\}, \{\lnot S\}\}\}, \underline{\{\forall h.\{\{\lnot W\}\}\}}\}\}$};
        \draw[->] (S1) to[bend right=20] node[pos=0.91, right] {$\mathrm{A1}^+$ (ex-5')} (S5);
        \node [node] at (0, -10) (S6) {${\mathcal{S}_6}'$};
        \node [right=0cm of S6] (S6_t) {$= \{\{\{A\}, \underline{\{\exists h.\{\{\lnot W\}, \{L\}, \{\lnot S\}\}\}}\}\}$};
        \draw[->] (S5) -- (S6) node [midway, right] {$\mathrm{A2}^+$ (ex-6')};
        \node [node] at (0, -11.5) (S7) {${\mathcal{S}_7}'$};
        \node [right=0cm of S7] (S7_t) {$= \{\{\{A\}\}, \underline{\{\{\lnot W\}, \{L\}, \{\lnot S\}\}}\}$};
        \draw[->] (S6) -- (S7) node [midway, right] {A3 (ex-7')};
        \node [below=0cm of S7] (S7_SAT) {SAT};
    \end{tikzpicture}
    \caption{An efficient derivation tree with $\mathrm{A1}^+$, $\mathrm{A2}^+$, and A3}
    \label{fig:solve_example_A1+}
\end{figure}

\begin{remark}
    Inference rule A1 selects a concept literal $L \in CL$ in a clause $CL$ but does not handle the concept literal $L$ or the complementary literal $\overline{L}$ in other clauses.
    This process may cause redundant derivation steps or backtracking.
    On the other hand, inference rule $\mathrm{A1}^+$ selects a concept literal $L \in CL$ in a clause $CL$ and simultaneously processes all clauses containing $L$ or $\overline{L'}$.
    As a result of this rule, every clause containing $L$ is transformed into the unit clause $\{L\}$ without any contradiction.
    In addition, the complementary literal $\overline{L}$ is removed from every clause containing $\overline{L}$ and is then not selected in later derivation steps.
    Therefore, $\mathrm{A1}^+$ is more efficient than A1.
\end{remark}
    
\begin{remark}
    Inference rule A2 can be applied to any clause (not limited to a unit clause), thus expanding the choice of inference rules in the early stages.
    If A2 selects a concept literal $L$, then the selection is equivalent to applying inference rule A1 to $L$ beforehand.
    By restricting the applications of A2, inference rule $\mathrm{A2}^+$ can be applied for a unit clause after inference rule $\mathrm{A1}^+$ is applied.
    This reduces the complexity of derivation because the selection of a concept literal in each clause is minimized by applying inference rule $\mathrm{A1}^+$.
    
\end{remark}

\section{Conclusion}

In this paper, we formalized CNF concepts in a conjunctive normal form for the description logic $\mathcal{ALC}$ where any $\mathcal{ALC}$ concept can be transformed to a CNF concept.
To decide the satisfiability of a CNF concept, we designed a decidable reasoning algorithm for clause sets in $\mathcal{ALC}$.
In particular, inference rules A1, A2, and A3 in our proposed reasoning algorithm provide efficient derivation steps owing to the clausal form of $\mathcal{ALC}$ concepts.
Furthermore, to improve the efficiency of the reasoning algorithm, inference rules A1 and A2 were improved to $\mathrm{A1}^+$ and $\mathrm{A2}^+$, thus reducing further derivation steps.
By formalizing  (restricted) CNF tableaux based on the semantics of CNF concepts, we proved the termination, soundness, and completeness of the two reasoning algorithms using A1, A2, and A3 as well as $\mathrm{A1}^+$, $\mathrm{A2}^+$, and A3, respectively.
The theoretical results are expected to lead to some applications of the proposed techniques for clausal reasoning to description logics---such as the resolution principle and solvers for the Boolean satisfiability problem (SAT solvers).

Our future work will include the formalization of a conjunctive normal form and its reasoning algorithm for more expressive description logics corresponding to OWL.
In addition, the proposed conjunctive normal form in $\mathcal{ALC}$ will be applied to fast SAT solvers to implement the reasoning algorithm for ontologies and knowledge bases.

\bibliographystyle{unsrt}
{
\small
\bibliography{bibliography}
}

\end{document}